\documentclass[11pt]{article}
\usepackage[T1]{fontenc}
\usepackage[utf8]{inputenc}
\usepackage{lmodern}
\usepackage[margin=1in]{geometry}
\usepackage{amsmath,amssymb,amsthm,mathtools}
\usepackage{graphicx,booktabs,microtype}
\usepackage{caption}
\usepackage{placeins}

\usepackage{cite}
\usepackage[colorlinks=true,linkcolor=blue,citecolor=blue,urlcolor=blue]{hyperref}
\hypersetup{pdftitle={Edge universality in Floquet sideband spectra},pdfauthor={Miguel Tierz},pdfsubject={Revised manuscript, 8 September 2026}}
\numberwithin{equation}{section}
\DeclareMathOperator{\Ai}{Ai}

\DeclareMathOperator{\sgn}{sgn}
\DeclareMathOperator{\sech}{sech}
\DeclareMathOperator{\csch}{csch}

\newcommand{\R}{\mathbb R}
\newcommand{\Z}{\mathbb Z}

\newcommand{\dd}{\,d}
\newcommand{\one}{\mathbf 1}
\newcommand{\KA}{K_{\rm Ai}}
\newtheorem{theorem}{Theorem}[section]
\newtheorem{proposition}[theorem]{Proposition}

\theoremstyle{remark}

\title{Edge universality in Floquet sideband spectra}
\author{Miguel Tierz\\[4pt]
{\normalsize\itshape Shanghai Institute for Mathematics and Interdisciplinary Sciences (SIMIS),}\\
{\normalsize\itshape Block A, No.~657 Songhu Road, Yangpu District, Shanghai 200438, China}\\[2pt]
{\normalsize\ttfamily tierz@simis.cn}}
\date{}

\begin{document}
\maketitle

\begin{abstract}
A monochromatic voltage drive redistributes noninteracting fermions over sidebands with Bessel amplitudes. A sharp Fermi step gives the discrete Bessel source kernel, whose large amplitude edge has the Airy limit known from random matrices and propagating fronts. At zero temperature and for energy independent transmission, the deficit of the dc bias derivative of low frequency partition noise from its large bias slope is proportional to the kernel diagonal, with ac amplitude fixed. Shifting the integer threshold coordinate down by half a sideband removes the leading finite amplitude correction. For reservoirs at a common temperature, we derive the rescaled deficit's crossover through one thermal convolution per reservoir. Its leading low temperature correction is twice that of the rescaled source occupation. We specify the geometry and mode resolution needed to distinguish this signature from the full Airy point process.
\end{abstract}

\section{Introduction}\label{sec:intro}

A voltage that varies periodically in time can create electron and hole excitations even when its time average vanishes. In the regime where the voltage acts through a phase on an incoming electronic channel, a monochromatic drive produces sidebands separated by the photon energy $\hbar\Omega$, with probabilities $J_n(A)^2$. Here $\Omega$ is the angular frequency, $V_{\rm ac}$ the voltage amplitude, $J_n$ the Bessel function of the first kind of order $n$, and $A=eV_{\rm ac}/\hbar\Omega$ the dimensionless drive amplitude, with $e>0$~\cite{TienGordon1963,PlateroAguado2004,PedersenBuett1998}. For an energy independent transmission probability the mean current is insensitive to this redistribution. The partition noise retains it~\cite{LesovikLevitov1994,BlanterBuettiker2000}, as demonstrated in conductors driven by microwave voltages~\cite{Schoelkopf1998,Reydellet2003}.

The transition between appreciable and exponentially small sideband weights occurs near $n=A$. Its width grows as $A^{1/3}$. The local approximation of an individual Bessel function by the Airy function $\Ai$ is classical~\cite{Watson1944,Olver1974,DLMF}. A filled Fermi sea introduces an additional structure: its outgoing correlation matrix contains a sum of products of Bessel amplitudes. At a sharp Fermi step this sum is the discrete Bessel kernel. Its Airy limit and determinantal statistics are established in the theory of random partitions~\cite{Johansson2001,Borodin2000} and in the propagation of a domain wall of free fermions~\cite{EislerRacz2013}. The latter system has, up to diagonal phases, exactly the same correlation matrix as the driven source considered here.

Our purpose is to use this correspondence to obtain a quantitative transport prediction, including the leading effect of finite drive amplitude and the crossover caused by temperature. The central observable is $\partial_V S_I$, the derivative of the low frequency current noise $S_I$ with respect to dc bias $V$ at fixed ac amplitude. For one spinless channel with transmission probability $0<\tau<1$, define the bias and edge scales and the large positive bias slope by
\begin{equation}\label{eq:scales-intro}
 x=\frac{eV}{\hbar\Omega},\qquad
 \kappa_A=\left(\frac A2\right)^{1/3},\qquad
 S'_\infty=\frac{2e^3}{h}\tau(1-\tau).
\end{equation}
At zero temperature and noninteger $x>0$, the normalized deficit obeys the exact identity
\begin{equation}\label{eq:main-deficit}
 \frac{S'_\infty-\partial_V S_I}{2S'_\infty}
 =\sum_{n\geq\lceil x\rceil}J_n(A)^2.
\end{equation}
The right hand side is a diagonal entry of the discrete Bessel kernel. Consequently, for $x=A+\kappa_A s$,
\begin{equation}\label{eq:main-limit}
 \frac{\kappa_A}{2S'_\infty}
       \bigl(S'_\infty-\partial_V S_I\bigr)
 \longrightarrow \KA(s,s)
 =\Ai'(s)^2-s\Ai(s)^2.
\end{equation}
Here $\KA$ is the Airy kernel, defined in \eqref{eq:airy-kernel}. Equations~\eqref{eq:main-deficit} and \eqref{eq:main-limit} follow from the standard noise formula and the known kernel limit. Their derivations and assumptions are given in Sections~\ref{sec:kernel}--\ref{sec:noise}. The limiting curve is the cumulative weight $\int_s^\infty\Ai(v)^2\dd v$. Individual sideband weights instead approach $\Ai(s)^2$ after multiplication by $\kappa_A^2$.

Two refinements are useful for interpreting a measurement. First, the sum in \eqref{eq:main-deficit} starts at an integer threshold. Its leading correction is the boundary term in the Euler--Maclaurin formula. Subsection~\ref{sec:correction} shows that using the coordinate $(\lceil x\rceil-A-1/2)/\kappa_A$ removes that correction. At the center of each noise slope plateau this coordinate is simply $(x-A)/\kappa_A$. Second, finite temperature smooths the occupation of the driven source once, whereas partition noise combines the occupations of two reservoirs. We derive an explicit noise crossover in Subsection~\ref{sec:temperature} and distinguish it from the finite temperature Airy kernel~\cite{Johansson2007,DeanLeDoussal2015,DeanLeDoussal2016}.

The physical model is a phase drive applied to one incoming lead relative to an undriven reference, followed by an elastic beam splitter~\cite{PedersenBuett1998,RychkovPolianskiBuett2005,MoskaletsBook}. The exact source kernel requires noninteracting fermions and a sharp incoming Fermi step. The noise formula further requires transmission to be approximately constant throughout the energies sampled by the drive and bias. A potential common to the entire apparatus is a gauge change and cannot generate the proposed excess noise. Section~\ref{sec:measurement} explains how the driven source enters a measured outgoing lead, and what changes when several Fermi steps or several quasienergy modes contribute.

This distinction also fixes the scope of the random matrix interpretation. The occupation process of a specified quasienergy sector has an Airy limit, including its higher correlations and gap probabilities. A noise slope measurement tests its diagonal scaling function. Establishing the full process experimentally requires additional information about correlations and the selected modes~\cite{LevitovLeeLesovik1996,Jullien2014,Roussel2017}. Finally, Section~\ref{sec:waveforms} examines what the geometry of the stationary points permits for other waveforms, without assuming a Pearcey kernel from the presence of a second drive tone.

\section{The sideband kernel and its Airy limit}\label{sec:kernel}

The phase drive determines both the sideband covariance and the joint statistics of occupied modes.

\subsection{Source geometry and conventions}

Write $\mathcal P=2\pi/\Omega$ for the drive period. A voltage $V_{\rm ac}\cos(\Omega t)$ applied to the source gives the phase $\phi(t)=A\sin(\Omega t)$. We use the convention
\begin{equation}\label{eq:phase}
 c(t)=e^{-i\phi(t)}a(t),\qquad
 e^{-iA\sin(\Omega t)}=\sum_{l\in\Z}J_l(A)e^{-il\Omega t}.
\end{equation}
Reversing the voltage or shifting its time origin changes removable sideband phases and leaves the monochromatic occupation and noise formulas below unchanged.

Choose a quasienergy $E\in[0,\hbar\Omega)$ and define $E_m=E+m\hbar\Omega$. The incoming and driven source modes are related by
\begin{equation}\label{eq:source-map}
 c_n(E)=\sum_{m\in\Z}J_{n-m}(A)a_m(E).
\end{equation}
The Bessel matrix $J_{nm}=J_{n-m}(A)$ is unitary on $\ell^2(\Z)$ because multiplication by the phase in \eqref{eq:phase} is unitary. Equivalently,
\begin{equation}\label{eq:orthogonality}
 \sum_{r\in\Z}J_{r+n}(A)J_{r+p}(A)=\delta_{np}.
\end{equation}

For precision, the covariance $C_{np}(E)$ denotes the coefficient of the energy delta function:
\begin{equation}\label{eq:covariance-convention}
 \langle c_p^\dagger(E')c_n(E)\rangle
 =\delta(E-E')C_{np}(E),\qquad
 C(E)=J N(E)J^\dagger,
\end{equation}
with $N_{mp}(E)=\delta_{mp}f(E_m)$, where $f$ is the incoming source Fermi occupation. This uses fields normalized in energy with anticommutator $\delta(E-E')$. It avoids assigning a finite occupation to an unsmeared continuum operator at an exactly specified energy. Joint occupation statistics can equivalently be defined with discrete, normalized modes in a finite observation time. Subsection~\ref{sec:modes} describes that regularization.

At zero temperature, with source chemical potential $\mu$, set
\begin{equation}\label{eq:fermi-index}
 m_F(E)=\left\lfloor\frac{\mu-E}{\hbar\Omega}\right\rfloor,
 \qquad \nu=n-m_F(E),\quad \lambda=p-m_F(E).
\end{equation}
The convention for a mode exactly at $\mu$ is immaterial in the continuum. For the discrete statements we take it to be filled. Since $f(E_m)=\one_{m\leq m_F}$, changing the summation variable to $r=m_F-m$ yields
\begin{equation}\label{eq:bessel-kernel}
 C_{np}(E)=B_A(\nu,\lambda),\qquad
 B_A(\nu,\lambda)=\sum_{r=0}^{\infty}J_{\nu+r}(A)J_{\lambda+r}(A).
\end{equation}
This identity holds for every $A\geq0$. The sum converges absolutely. It is the discrete Bessel kernel with an endpoint convention in which the first summand is $r=0$~\cite{Johansson2001,Borodin2000}.

\subsection{Occupation statistics and exact identities}

The incoming Fermi sea is a Slater determinant, and \eqref{eq:source-map} preserves that structure. Thus $B_A=JPJ^\dagger$ is an orthogonal projection, where $P$ projects onto the filled incoming modes. Wick's theorem gives, for distinct normalized sideband modes at one quasienergy,
\begin{equation}\label{eq:dpp}
 \left\langle\prod_{j=1}^k N_{\nu_j}\right\rangle
 =\det[B_A(\nu_i,\nu_j)]_{i,j=1}^k,
 \qquad N_\nu=c_{m_F+\nu}^\dagger c_{m_F+\nu}.
\end{equation}
This determinant formula defines a determinantal point process of occupied sideband indices. Randomness refers to occupation measurements on the fixed sideband energies. In particular,
\begin{equation}\label{eq:covariance}
 \langle N_\nu\rangle=B_A(\nu,\nu),\qquad
 \langle N_\nu N_\lambda\rangle_c=-|B_A(\nu,\lambda)|^2
 \quad(\nu\ne\lambda).
\end{equation}
Here $\langle N_\nu N_\lambda\rangle_c=\langle N_\nu N_\lambda\rangle-\langle N_\nu\rangle\langle N_\lambda\rangle$. For any finite set $I\subset\Z$, its particle number $N_I=\sum_{\nu\in I}N_\nu$ has generating function
\begin{equation}\label{eq:fcs}
 \langle z^{N_I}\rangle
 =\det\bigl[1+(z-1)B_A|_I\bigr].
\end{equation}
These standard formulas for a number-conserving fermionic Gaussian state also apply when the covariance is a positive contraction, with eigenvalues between zero and one, rather than a projection~\cite{FetterWalecka,HKPV2009,LevitovLeeLesovik1996}. Thermal states and the reduced state of one outgoing lead are examples. Such mixed states need not be Slater determinants.

For $A>0$, the recurrence $2nJ_n(A)/A=J_{n-1}(A)+J_{n+1}(A)$ gives the useful identity
\begin{equation}\label{eq:cd}
 B_A(\nu,\lambda)
 =\frac A2\,
 \frac{J_{\nu-1}(A)J_\lambda(A)-J_\nu(A)J_{\lambda-1}(A)}{\nu-\lambda},
 \qquad \nu\ne\lambda.
\end{equation}
Indeed, multiplying the recurrence at $\nu+r$ by $J_{\lambda+r}$ and subtracting the recurrence at $\lambda+r$ times $J_{\nu+r}$ produces a telescoping sum. Its lower endpoint is the numerator in \eqref{eq:cd}, and its upper endpoint vanishes. On the diagonal we retain the positive sum in \eqref{eq:bessel-kernel}. It satisfies
\begin{equation}\label{eq:diagonal-difference}
 B_A(\nu,\nu)-B_A(\nu+1,\nu+1)=J_\nu(A)^2.
\end{equation}

The same kernel describes the front generated by a fully occupied half-chain evolving with nearest-neighbor hopping. The single-particle propagator in that problem has matrix elements $i^{n-m}J_{n-m}(t)$, up to a convention for the hopping sign. Removing diagonal phases and setting $t=A$ gives \eqref{eq:bessel-kernel}. Its Airy counting statistics are therefore directly relevant here~\cite{EislerRacz2013}. The physical coordinate has changed from position along a chain to the sideband index of a driven source.

\subsection{The Airy limit and its finite amplitude correction}\label{sec:airy}\label{sec:correction}

Define
\begin{equation}\label{eq:airy-kernel}
 \KA(s,t)=\int_0^\infty\Ai(s+u)\Ai(t+u)\dd u
 =\frac{\Ai(s)\Ai'(t)-\Ai'(s)\Ai(t)}{s-t},
\end{equation}
where the second expression is for $s\ne t$. The diagonal and its derivative are
\begin{equation}\label{eq:D}
 D_0(s):=\KA(s,s)=\Ai'(s)^2-s\Ai(s)^2
 =\int_s^\infty\Ai(v)^2\dd v,\qquad D_0'(s)=-\Ai(s)^2.
\end{equation}
This is the correlation kernel at the soft edge of the Gaussian unitary ensemble~\cite{TracyWidom1994,Forrester2010}.

\begin{theorem}[Airy limit]\label{thm:airy}
Let $\kappa_A=(A/2)^{1/3}$ and $\nu_A(s)=\lfloor A+\kappa_A s\rfloor$. Then
\begin{equation}\label{eq:kernel-limit}
 \kappa_A B_A\bigl(\nu_A(s),\nu_A(t)\bigr)
 \longrightarrow\KA(s,t)
\end{equation}
locally uniformly on $\R^2$ as $A\to\infty$. After embedding each lattice site as a cell of width $1/\kappa_A$ on the real line, the kernels restricted to a right half-line converge in trace norm. In particular, if $N_A=\lceil A+\kappa_A s_0\rceil$,
\begin{equation}\label{eq:gap}
 \det\bigl[1-B_A|_{\{N_A,N_A+1,\ldots\}}\bigr]
 \longrightarrow
 F_2(s_0):=\det\bigl[1-\KA|_{[s_0,\infty)}\bigr].
\end{equation}
\end{theorem}

The kernel limit is known~\cite{Johansson2001,EislerRacz2013}. Appendix~\ref{app:airy} gives a proof with the endpoint and normalization used here, including the domination needed for the half-line determinant. The mechanism can be seen directly from the turning-point approximation
\begin{equation}\label{eq:airy-amplitude}
 \kappa_A J_n(A)
 =\Ai\!\left(\frac{n-A}{\kappa_A}\right)+O(\kappa_A^{-2}),
\end{equation}
uniformly when the displayed argument stays in a compact set~\cite{Olver1974,DLMF}. Multiplying \eqref{eq:bessel-kernel} by $\kappa_A$ turns it into a Riemann sum with spacing $\kappa_A^{-1}$ and limit \eqref{eq:airy-kernel}. The floor in $\nu_A(s)$ adds an error of order $\kappa_A^{-1}$ to the Airy argument. Thus the finite lattice correction must be distinguished from the error in \eqref{eq:airy-amplitude}.

At a fixed quasienergy, \eqref{eq:covariance} and \eqref{eq:kernel-limit} give
\begin{equation}\label{eq:scaled-covariance}
 \kappa_A\langle N_{\nu_A(s)}\rangle\to D_0(s),\qquad
 \kappa_A^2\langle N_{\nu_A(s)}N_{\nu_A(t)}\rangle_c
 \to-|\KA(s,t)|^2\quad(s\ne t).
\end{equation}
The determinants in \eqref{eq:gap} are Fredholm determinants of the kernels restricted to the indicated sets. They give the probability of finding no occupied sideband at or above the threshold. Its limit is the Tracy--Widom distribution $F_2$~\cite{TracyWidom1994}. This conclusion concerns a specified collection of normalized modes. It is not a statement about arbitrary counts integrated over energy.

The endpoint of the Bessel sum determines the leading correction at finite drive amplitude.

\begin{proposition}[Leading correction]\label{prop:correction}
For integer $\nu,\lambda$, put $s_\nu=(\nu-A)/\kappa_A$ and $s_\lambda=(\lambda-A)/\kappa_A$. Uniformly when these two variables lie in a fixed compact set,
\begin{align}
 \kappa_A B_A(\nu,\lambda)
 &=\KA(s_\nu,s_\lambda)
   +\frac{\Ai(s_\nu)\Ai(s_\lambda)}{2\kappa_A}
   +O(\kappa_A^{-2})\label{eq:first-correction}\\
 &=\KA\!\left(s_\nu-\frac1{2\kappa_A},
                 s_\lambda-\frac1{2\kappa_A}\right)
   +O(\kappa_A^{-2}).\label{eq:midpoint}
\end{align}
In particular,
\begin{equation}\label{eq:diagonal-midpoint}
 \kappa_A B_A(\nu,\nu)
 =D_0\!\left(\frac{\nu-A-1/2}{\kappa_A}\right)
 +O(\kappa_A^{-2}).
\end{equation}
\end{proposition}

\begin{proof}
The integrable remainder estimate proved in Appendix~\ref{app:estimate} permits \eqref{eq:airy-amplitude} to be summed with total error $O(\kappa_A^{-2})$ after rescaling. For $F(u)=\Ai(s_\nu+u)\Ai(s_\lambda+u)$, the Euler--Maclaurin formula gives
\[
 \frac1{\kappa_A}\sum_{r=0}^\infty F(r/\kappa_A)
 =\int_0^\infty F(u)\dd u+\frac{F(0)}{2\kappa_A}
 +O(\kappa_A^{-2}).
\]
This proves \eqref{eq:first-correction}. Differentiating the integral in \eqref{eq:airy-kernel} gives
$(\partial_s+\partial_t)\KA(s,t)=-\Ai(s)\Ai(t)$, which proves \eqref{eq:midpoint} by Taylor expansion.
\end{proof}

The rank-one term in \eqref{eq:first-correction} already appears in Eq.~(10) of Ref.~\cite{EislerRacz2013}. Here its role is to fix the bias coordinate in a finite amplitude comparison. Equation~\eqref{eq:diagonal-midpoint} uses actual integer thresholds. It does not interpolate the exact noise staircase. At arbitrary values of the continuous bias, rounding still matters at order $\kappa_A^{-1}$.

\begin{figure}[!tbp]
 \centering
 \includegraphics[width=\linewidth]{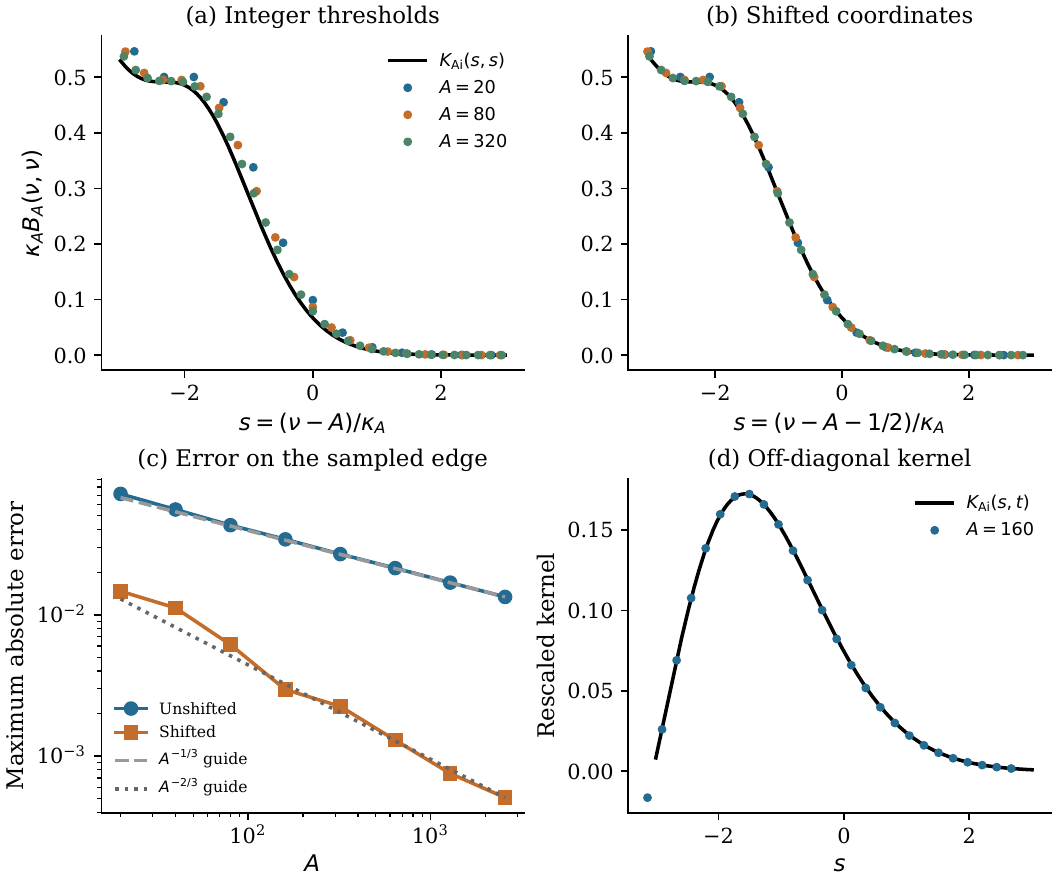}
 \caption{Exact Bessel sums and their Airy approximation. (a)~The diagonal plotted at integer thresholds. (b)~The same values plotted after shifting the coordinate by half a sideband. (c)~Maximum absolute errors on the integer sites satisfying $|\nu-A|\leq3\kappa_A$, for eight amplitudes from $20$ to $2560$. The dashed and dotted lines show the predicted powers, with their amplitudes fixed by the last data point. They are guides, not fitted error bounds. (d)~An off-diagonal comparison at $A=160$ and $\lambda=160$, using the shifted coordinates of Proposition~\ref{prop:correction}. Thus $t=-1/(2\kappa_A)$. All points are calculated from \eqref{eq:bessel-kernel}.}
 \label{fig:edge}
\end{figure}

Figure~\ref{fig:edge} illustrates the finite amplitude correction, and Table~\ref{tab:errors} quantifies the change in the diagonal approximation. The errors compare the same exact sum with $D_0((\nu-A)/\kappa_A)$ and with $D_0((\nu-A-1/2)/\kappa_A)$, respectively. The number of sampled integer sites is given explicitly. These are absolute errors of the rescaled observable, not relative errors or bounds throughout a continuous interval. Appendix~\ref{app:numerics} describes the numerical evaluation.
\begin{table}[ht]
 \centering
 \begin{tabular}{rrrr}
\toprule
$A$ & Sites & Unshifted error & Shifted error \\
\midrule
20 & 13 & $7.152\times10^{-2}$ & $1.465\times10^{-2}$ \\
80 & 21 & $4.297\times10^{-2}$ & $6.144\times10^{-3}$ \\
320 & 33 & $2.683\times10^{-2}$ & $2.249\times10^{-3}$ \\
1280 & 51 & $1.687\times10^{-2}$ & $7.542\times10^{-4}$ \\
\bottomrule
\end{tabular}

 \caption{Maximum errors for the exact diagonal on $|\nu-A|\leq3\kappa_A$. The shift removes the $O(A^{-1/3})$ term and leaves an $O(A^{-2/3})$ error.}
 \label{tab:errors}
\end{table}

\FloatBarrier
\section{Photo-assisted partition noise}\label{sec:noise}

Consider a static, elastic beam splitter with one incoming lead driven as in \eqref{eq:phase}, an undriven reference lead, and dc chemical potential difference $eV$. We assume a single spinless channel and a transmission probability $0<\tau<1$ that is independent of energy over the full sampled window. With $\delta I(t)=I(t)-\langle I(t)\rangle$, we use the symmetrized zero frequency convention
\begin{equation}\label{eq:noise-convention}
 S_I=\frac1{\mathcal P}\int_0^{\mathcal P}\dd t_0
       \int_{-\infty}^{\infty}\dd t\,
 \bigl\langle\{\delta I(t_0+t),\delta I(t_0)\}\bigr\rangle.
\end{equation}
With this convention a Poissonian current has noise $2e|I|$.

\subsection{Noise at zero temperature}

At zero temperature the standard photo-assisted noise is~\cite{LesovikLevitov1994,PedersenBuett1998,BlanterBuettiker2000}
\begin{equation}\label{eq:noise-zero}
 S_I(V,0)=\frac{2e^2}{h}\tau(1-\tau)
             \sum_{l\in\Z}J_l(A)^2|eV+l\hbar\Omega|.
\end{equation}
The derivative below always holds $A$, $\Omega$, $\tau$, and the electron temperature fixed. Varying the ac amplitude together with the dc voltage would add a different derivative term.

For noninteger $x=eV/\hbar\Omega>0$, differentiation gives
\begin{equation}\label{eq:noise-slope}
 \frac{\partial_V S_I}{S'_\infty}
 =\sum_lJ_l(A)^2\sgn(x+l)
 =1-2B_A(\lceil x\rceil,\lceil x\rceil).
\end{equation}
Here we used $J_{-l}(A)^2=J_l(A)^2$ and $\sum_lJ_l(A)^2=1$. This proves \eqref{eq:main-deficit}. The high bias slope $S'_\infty$ is approached asymptotically. Bessel weights have no exactly finite support. At an integer $x=L$ the derivative has one-sided values. Its jump is
\begin{equation}\label{eq:jump}
 \bigl[\partial_V S_I\bigr]_{x=L+}
 -\bigl[\partial_V S_I\bigr]_{x=L-}
 =2S'_\infty J_L(A)^2.
\end{equation}
Assigning the mean of these two values is a convention, not an additional plateau. Since the noise is even in $V$, the positive-bias comparison applies at $V<0$ with $|V|$ as the bias and $-\partial_V S_I$ as the slope.

For noninteger $x>0$, let
\begin{equation}\label{eq:normalized-deficit}
 d_{A,0}(x)=\frac{S'_\infty-\partial_V S_I(V,0)}{2S'_\infty},
 \qquad \mathcal D_{A,0}(s)=\kappa_A d_{A,0}(A+\kappa_A s).
\end{equation}
Theorem~\ref{thm:airy} implies \eqref{eq:main-limit}. For a finite amplitude comparison the stronger statement is
\begin{equation}\label{eq:noise-shift}
 \kappa_A d_{A,0}(x)
 =D_0\!\left(\frac{\lceil x\rceil-A-1/2}{\kappa_A}\right)
 +O(\kappa_A^{-2}),
\end{equation}
uniformly when the argument remains bounded. If $x=L-1/2$ is the center of the noise slope plateau with $\lceil x\rceil=L$, the argument in \eqref{eq:noise-shift} equals $s=(x-A)/\kappa_A$. This gives a direct sampling prescription without subtracting a theoretical correction from the data.

The distinction in Figure~\ref{fig:noise}(a) follows from the exact difference identity \eqref{eq:diagonal-difference}. A distributional second derivative of the noise gives peaks with weights proportional to $J_L(A)^2$, whereas the deficit of its first derivative gives their cumulative sum. The latter approaches the mean density of the Airy point process after rescaling. It can be inferred from the noise formula without measuring the off-diagonal kernel, so agreement with \eqref{eq:main-limit} alone does not determine the higher occupation correlations.

\begin{figure}[!hbp]
 \centering
 \includegraphics[width=\linewidth]{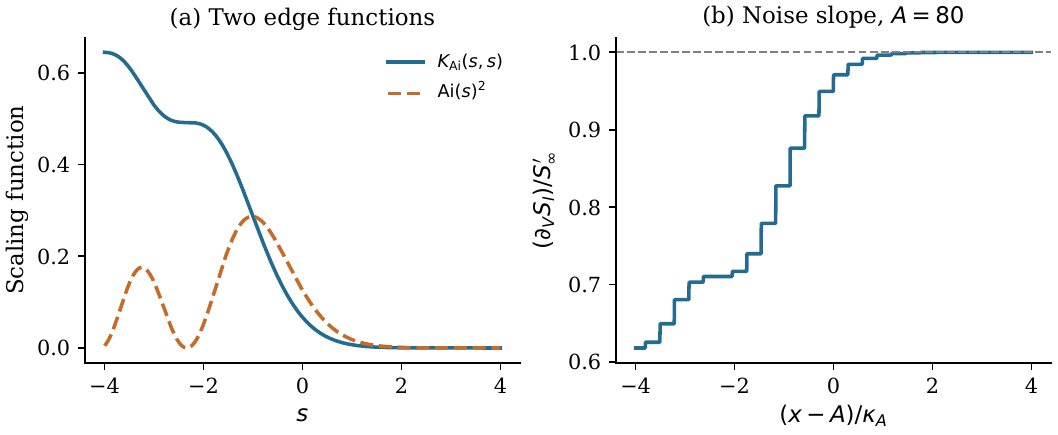}
 \caption{The observable selected by the bias derivative. (a)~The rescaled noise slope deficit $\kappa_A d_{A,0}$ approaches $D_0(s)=\int_s^\infty\Ai(v)^2\dd v$, while the rescaled individual sideband weights $\kappa_A^2J_n(A)^2$ approach $\Ai(s)^2$. Their relation is $D_0'=-\Ai^2$. (b)~The exact zero temperature noise slope at $A=80$, normalized by its limiting value. Its plateau values follow \eqref{eq:noise-slope}. The dashed line marks $S'_\infty$.}
 \label{fig:noise}
\end{figure}

For several independent channels with the same applied voltage, replace $\tau(1-\tau)$ by $\sum_a\tau_a(1-\tau_a)$ in $S'_\infty$ and \eqref{eq:noise-zero}. This includes spin degeneracy when each spin channel is counted separately. Normalizing by the measured high bias slope removes this overall factor, provided the channel transmissions remain energy independent. The mean current is $I=(e^2/h)(\sum_a\tau_a)V$. Its lack of sideband structure in this model is consistent with the established photo-assisted noise experiments~\cite{Schoelkopf1998,Reydellet2003}.

\subsection{Source occupations at finite temperature}

We now keep both reservoirs at a common electron temperature $T$. Two dimensionless temperatures will be useful:
\begin{equation}\label{eq:temperatures}
 \vartheta=\frac{k_BT}{\hbar\Omega},\qquad
 \theta_A=\frac{\vartheta}{\kappa_A}.
\end{equation}
The parameter $\vartheta$ controls the rounding of individual sideband thresholds. The parameter $\theta_A$ compares this rounding with the full Airy edge width.

Let $\eta(E)=(\mu-E)/\hbar\Omega-m_F(E)\in[0,1)$. With the Fermi distribution in \eqref{eq:covariance-convention}, the source kernel is
\begin{equation}\label{eq:source-thermal}
 B_{A,\vartheta}^{(\eta)}(\nu,\lambda)
 =\sum_{r\in\Z}J_{\nu+r}(A)J_{\lambda+r}(A)
          \frac1{1+e^{-(r+\eta)/\vartheta}}.
\end{equation}
This covariance is a positive contraction. As $A\to\infty$ with $\theta_A\to\theta>0$, the edge limit is
\begin{equation}\label{eq:source-thermal-limit}
 \kappa_A B_{A,\vartheta}^{(\eta)}(\nu_A(s),\nu_A(t))
 \longrightarrow
 K_\theta^{\rm src}(s,t)
 :=\int_{-\infty}^{\infty}
       \frac{\Ai(s+u)\Ai(t+u)}{1+e^{-u/\theta}}\dd u.
\end{equation}
The fractional offset $\eta$ disappears at this scale. This is the familiar finite temperature Airy kernel, given in Eq.~(14) of Ref.~\cite{Johansson2007}, Eq.~(10) of Ref.~\cite{DeanLeDoussal2015}, and Eq.~(134) of Ref.~\cite{DeanLeDoussal2016}. Its zero temperature limit is \eqref{eq:airy-kernel}. The proof by Riemann sums is the same as at zero temperature, with exponential Fermi decay controlling the negative-$u$ tail. Details of the required bounds appear in Appendix~\ref{app:thermal-limit}. The same discrete finite temperature Bessel kernel and its Airy edge limit also arise in the periodic Schur process~\cite{BeteaBouttier2019}.

Introduce the even probability density
\begin{equation}\label{eq:logistic}
 q_\theta(v)=\frac1{4\theta}\sech^2\!\left(\frac{v}{2\theta}\right),
 \qquad \int_\R q_\theta(v)\dd v=1.
\end{equation}
On the diagonal, integration in the Fermi factor gives the convolution identity
\begin{equation}\label{eq:source-convolution}
 K_\theta^{\rm src}(s,s)=(q_\theta*D_0)(s),\qquad
 (f*g)(s):=\int_\R f(v)g(s-v)\dd v.
\end{equation}
Only one reservoir occupation has entered this formula.

\subsection{The thermal broadening of partition noise}\label{sec:temperature}

For the geometry specified above, the finite temperature noise is~\cite{BlanterBuettiker2000,Reydellet2003,MoskaletsBook}
\begin{equation}\label{eq:noise-thermal}
 S_I(V,T)=\frac{2e^2}{h}\left[
 2k_BT\tau^2+\tau(1-\tau)\sum_{l\in\Z}J_l(A)^2
 H_T(eV+l\hbar\Omega)\right],
\end{equation}
where
\begin{equation}\label{eq:H}
 H_T(\epsilon)=\epsilon\coth\!\left(\frac{\epsilon}{2k_BT}\right),
 \qquad H_T(0)=2k_BT.
\end{equation}
At $A=V=0$, \eqref{eq:noise-thermal} gives the equilibrium noise $4k_BT(e^2\tau/h)$. Its first term is independent of bias under our assumptions and does not affect the slope deficit.

For $\vartheta>0$, define the smooth odd function
\begin{equation}\label{eq:g}
 g_\vartheta(y)=\coth\!\left(\frac{y}{2\vartheta}\right)
 -\frac{y}{2\vartheta}\csch^2\!\left(\frac{y}{2\vartheta}\right),
 \qquad g_\vartheta(0)=0.
\end{equation}
Differentiating \eqref{eq:noise-thermal} and using reflection of the Bessel weights gives
\begin{equation}\label{eq:thermal-deficit}
 d_{A,\vartheta}(x)
 :=\frac{S'_\infty-\partial_V S_I(V,T)}{2S'_\infty}
 =\frac12\sum_{n\in\Z}J_n(A)^2
             \bigl[1-g_\vartheta(x-n)\bigr].
\end{equation}
For the convolution below, define $d_{A,0}(x)=\sum_{n>x}J_n(A)^2$ for all real $x$, with any fixed convention at the integer jumps. This extends the positive-bias expression using the signed derivative $\partial_V S_I$.

\begin{proposition}[Exact noise broadening]\label{prop:thermal-convolution}
For every $A\geq0$ and $\vartheta>0$,
\begin{equation}\label{eq:exact-convolution}
 d_{A,\vartheta}=w_\vartheta*d_{A,0},\qquad
 w_\vartheta=q_\vartheta*q_\vartheta.
\end{equation}
The even probability density $w_\vartheta$ has the explicit form
\begin{equation}\label{eq:w}
 w_\vartheta(v)=\frac1{2\vartheta}
 \csch^2\!\left(\frac{v}{2\vartheta}\right)
 \left[\frac{v}{2\vartheta}\coth\!\left(\frac{v}{2\vartheta}\right)-1\right],
 \qquad w_\vartheta(0)=\frac1{6\vartheta}.
\end{equation}
\end{proposition}

Appendix~\ref{app:thermal-identity} proves this identity directly from two Fermi occupations. Each reservoir contributes a density $q_\vartheta$. The difference of their thermal energy shifts has density $q_\vartheta*q_\vartheta$. In particular, the smoothing function for partition noise is not $q_\vartheta$ itself.

\begin{theorem}[Noise crossover at the Airy edge]\label{thm:thermal-noise}
If $A\to\infty$ and $\vartheta/\kappa_A\to\theta>0$, then, locally uniformly in $s$,
\begin{equation}\label{eq:noise-crossover}
 \kappa_A d_{A,\vartheta}(A+\kappa_A s)
 \longrightarrow
 D_\theta^{\rm noise}(s)
 :=(w_\theta*D_0)(s)
 =(q_\theta*q_\theta*D_0)(s).
\end{equation}
An equivalent integral is
\begin{equation}\label{eq:noise-crossover-integral}
 D_\theta^{\rm noise}(s)
 =\frac12\int_{-\infty}^{\infty}\Ai(s+u)^2
       \bigl[1+g_\theta(u)\bigr]\dd u.
\end{equation}
Both expressions tend to $D_0(s)$ as $\theta\downarrow0$.
\end{theorem}

The proof follows from \eqref{eq:exact-convolution}, the Airy limit, and exponential decay of the broadening density. Appendix~\ref{app:thermal-limit} supplies a uniform bound that justifies the limit under the integral. Figure~\ref{fig:thermal} illustrates the finite temperature transport prediction: the source diagonal in \eqref{eq:source-convolution} is smoothed once more by the reference reservoir.

\begin{figure}[!hbp]
 \centering
 \includegraphics[width=\linewidth]{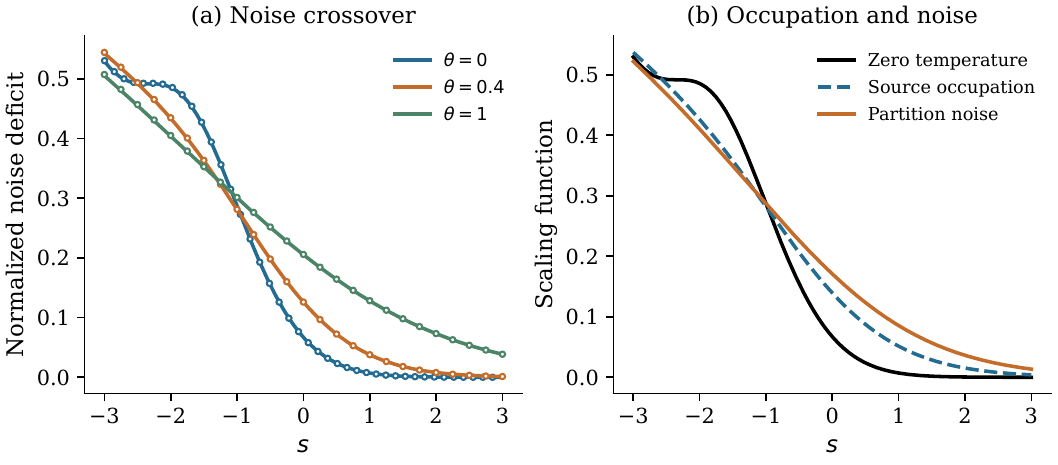}
 \caption{Thermal edge profiles. (a)~The noise crossover \eqref{eq:noise-crossover} for $\theta=0,0.4,1$ (curves), compared with exact Bessel sums at $A=640$ (open circles). The samples at positive temperature use $x=A+\kappa_A s$. The zero temperature points are taken at noise slope plateau centers. (b)~At the same $\theta=0.7$, the source occupation and the noise deficit have different broadening functions. The zero temperature diagonal is shown for reference.}
 \label{fig:thermal}
\end{figure}

The variance of $q_\theta$ is $\pi^2\theta^2/3$, so that of $w_\theta$ is $2\pi^2\theta^2/3$. Expanding the convolutions on compact $s$ intervals gives
\begin{align}
 K_\theta^{\rm src}(s,s)
 &=D_0(s)-\frac{\pi^2}{3}\theta^2\Ai(s)\Ai'(s)
     +O(\theta^4),\label{eq:source-sommerfeld}\\
 D_\theta^{\rm noise}(s)
 &=D_0(s)-\frac{2\pi^2}{3}\theta^2\Ai(s)\Ai'(s)
     +O(\theta^4).\label{eq:noise-sommerfeld}
\end{align}
These expansions are taken after the edge limit. Finite-$A$ lattice effects are separate. The source formula is the usual Sommerfeld correction. The noise coefficient is twice as large because both reservoirs contribute thermal fluctuations. Since $\Ai(0)\Ai'(0)=-\sqrt3/(6\pi)$, both corrections are positive at $s=0$ and neither vanishes there. The noise coefficient at $s=0$ is $\pi\sqrt3/9$.

The scale for appreciable broadening of the Airy envelope is
\begin{equation}\label{eq:thermal-scale}
 k_BT\sim\hbar\Omega\kappa_A.
\end{equation}
Resolving each photon threshold instead requires $k_BT\ll\hbar\Omega$. These are different requirements: a smooth noise curve can retain its Airy envelope after the individual steps have been rounded. For illustration, at $\Omega/2\pi=10\,\mathrm{GHz}$ and $A=50$, $\kappa_A\simeq2.92$, $\hbar\Omega/k_B\simeq0.480\,\mathrm K$, and the scale in \eqref{eq:thermal-scale} is about $1.40\,\mathrm K$. These values specify the model's scales rather than a demonstrated operating range. Heating caused by the drive or a temperature that varies with bias would have to be included in using \eqref{eq:thermal-deficit}.

\FloatBarrier
\section{Scattering, several Fermi steps, and detection}\label{sec:measurement}

The source kernel enters a measurement through the scattering geometry, the incoming occupations and the modes selected by the detector.

\subsection{The driven source in an outgoing lead}

The covariance \eqref{eq:bessel-kernel} belongs to the driven incoming source. To obtain a measured outgoing covariance one must specify which leads are driven. If $\mathcal U(\epsilon)$ is the static scattering matrix and only lead $L$ is driven, the exact composition is
\begin{equation}\label{eq:physical-scattering}
 b_{\alpha n}(E)=\mathcal U_{\alpha L}(E_n)c_{Ln}(E)
                  +\mathcal U_{\alpha R}(E_n)a_{Rn}(E).
\end{equation}
Thus the Floquet amplitudes from the two reservoirs are
\begin{equation}\label{eq:physical-floquet}
 S^{\alpha L}_{nm}(E)=\mathcal U_{\alpha L}(E_n)J_{n-m}(A),
 \qquad
 S^{\alpha R}_{nm}(E)=\mathcal U_{\alpha R}(E_n)\delta_{nm}.
\end{equation}
The general scattering identity is $C^{\rm out}=SNS^\dagger$~\cite{PedersenBuett1998,MoskBuett2002,MoskaletsBook}. Independent reservoirs give
\begin{align}\label{eq:output-covariance}
 C_{np}^{\rm out}(E)
 ={}&\mathcal U_{\alpha L}(E_n)\mathcal U_{\alpha L}(E_p)^*
          C_{np}^{\rm src}(E)\nonumber\\
 &+|\mathcal U_{\alpha R}(E_n)|^2\delta_{np}f_R(E_n).
\end{align}
For constant transmission and $[N_R]_{np}=\delta_{np}f_R(E_n)$, this reduces to
\begin{equation}\label{eq:one-drive-output}
 C^{\rm out}=\tau C^{\rm src}+(1-\tau)N_R,
\end{equation}
up to removable scattering phases. At zero temperature, if the reference modes are empty on the counted upper-edge half-line, the rescaled outgoing covariance tends to $\tau\KA$. This is a thinned Airy point process: each point is retained independently with probability $\tau$~\cite{HKPV2009}. Its limiting gap probability above $s_0$ is $\det(1-\tau\KA|_{[s_0,\infty)})$, not $F_2(s_0)$ unless $\tau=1$.

Equation~\eqref{eq:physical-floquet} also separates two assumptions that play different roles. The Bessel source amplitudes are fixed by the imposed phase and do not require a flat static scatterer. Energy dependence of the scatterer remains explicitly in \eqref{eq:output-covariance}. The simple noise expression \eqref{eq:noise-zero} requires flat transmission across the sampled energies. Smoothness only over one photon energy is insufficient when $A\gg1$: the ac window is of order $A\hbar\Omega$, and the dc bias and temperature enlarge the relevant range.

If identical time-dependent phases are applied to both incoming leads of an energy independent beam splitter, their relative ac phase vanishes and there is no excess photo-assisted partition noise. If the detector and all potentials are also shifted together, the change is entirely a gauge transformation. A common factor $J_{n-m}(A)$ for every incoming reservoir therefore does not describe the noise experiment with one driven lead of Section~\ref{sec:noise}.

\subsection{A source with two occupation steps}\label{sec:steps}

A different preparation is a number-conserving fermionic Gaussian source with diagonal incoming covariance $N_{mp}=\delta_{mp}f_m$ and mixing weight $0\leq\tau\leq1$, where
\begin{equation}\label{eq:two-step-f}
 f_m=\tau\one_{m\leq m_1}+(1-\tau)\one_{m\leq m_2},
 \qquad m_1>m_2,
\end{equation}
followed by a common phase drive. Its covariance is exactly
\begin{equation}\label{eq:two-step-kernel}
 C_{np}=\tau B_A(n-m_1,p-m_1)
 +(1-\tau)B_A(n-m_2,p-m_2).
\end{equation}
This is not the covariance \eqref{eq:one-drive-output}. The preparations differ. Formula~\eqref{eq:two-step-kernel} is useful independently of the noise geometry.

Set $\Delta=m_1-m_2$. Near the upper edge $n=m_1+A+\kappa_A s$, if $\Delta/\kappa_A\to\infty$, the second kernel is exponentially small and the rescaled covariance tends to $\tau\KA(s,t)$. If $\Delta/\kappa_A\to d<\infty$, the limit is instead
\begin{equation}\label{eq:overlap}
 \tau\KA(s,t)+(1-\tau)\KA(s+d,t+d).
\end{equation}
This is one determinantal covariance obtained from a partially occupied input. It is not the superposition of two independent point processes. The covariance of occupations contains the squared modulus of the whole expression.

When $\Delta/\kappa_A\to\infty$, the first kernel need not be negligible near $n=m_2+A$, the lower of the two upper sideband edges. The full covariance therefore need not have an isolated Airy limit there. An exact counterexample makes the size of the background clear. Take integer $A$, $\Delta=A$, and evaluate at $n=m_2+A=m_1$. The first term on the diagonal is
\begin{equation}\label{eq:lower-counterexample}
 \tau B_A(0,0)=\frac\tau2\bigl[1+J_0(A)^2\bigr]\longrightarrow\frac\tau2.
\end{equation}
Multiplication by $\kappa_A$ makes this background diverge for every fixed $\tau>0$, however small. A small fixed mixing weight therefore does not justify assigning a thinned Airy point process to the lower edge. Subtracting a mean background also does not remove its occupation correlations.

\subsection{Temporal selection and normalized modes}\label{sec:modes}

A periodic mean signal obtained from a number-conserving one-body observable can be written in the form
\begin{equation}\label{eq:mean-signal}
 y(t)=\int_0^{\hbar\Omega}\frac{\dd E}{h}
       \sum_{n,p}G_{np}(E)C_{np}(E)e^{i(p-n)\Omega t}.
\end{equation}
The coefficient $G_{np}$ specifies the observable and its coupling to the measured channel. The sums are understood for finite detector bandwidth or after the equilibrium subtraction appropriate to the observable. A terminal current also includes its incoming contribution. Equation~\eqref{eq:mean-signal} describes each such one-body contribution. This is a statement about the first moment. At the operator level a current also contains terms coupling distinct quasienergies. Those terms cannot be dropped in a noise calculation~\cite{BlanterBuettiker2000,MoskaletsBook}. Electronic coherence measurements and their signal-processing interpretation are discussed in Ref.~\cite{Roussel2017}.

For a rectangular gate of length $L=M\mathcal P$, where $M$ is a positive integer, and demodulation at an integer harmonic $k\Omega$ with $k\in\Z$,
\begin{equation}\label{eq:gate}
 \frac1L\int_{t_0}^{t_0+L}e^{-ik\Omega t}y(t)\dd t
 =\int_0^{\hbar\Omega}\frac{\dd E}{h}
       \sum_nG_{n,n+k}(E)C_{n,n+k}(E).
\end{equation}
Fourier orthogonality fixes the difference $p-n=k$. It leaves a sum over the individual sideband index and an integral over quasienergy. Thus even ideal gating does not by itself isolate the spectral edge. A nonrectangular window replaces the Kronecker selection by its Fourier response~\cite{Harris1978}. A general linear filter is represented by the resulting effective temporal weight. Off-diagonal first moments further require a stable phase reference to the drive.

One can make the occupation process at fixed quasienergy precise by first imposing an observation time $\mathcal T_{\rm obs}=M\mathcal P$ with periodic temporal boundary conditions. Energy modes are then spaced by $h/\mathcal T_{\rm obs}$, and there are $M$ quasienergy sectors in a Floquet zone. Each sector has normalized fermionic modes with occupations zero or one. The covariance is block diagonal in that sector index, so \eqref{eq:dpp}--\eqref{eq:fcs} apply separately to each block. The continuum scattering covariance is recovered as $M\to\infty$.

For a detector that counts modes in sets $I_j$ with finite expected occupation in each sector $j$, let $C^{(j)}$ denote the covariance in that sector and $N_{\rm det}$ the total count. Its generating function is
\begin{equation}\label{eq:many-modes}
 \langle z^{N_{\rm det}}\rangle
 =\prod_j\det\bigl[1+(z-1)C^{(j)}|_{I_j}\bigr].
\end{equation}
For fixed $M$, when the same threshold yields $M$ identical Airy blocks, the zero count probability tends to $F_2(s_0)^M$, rather than one copy of $F_2(s_0)$. Other selections acting separately within the quasienergy sectors give the corresponding determinant products. A narrow energy bin can still contain many modes if the observation time is long. The mode count is therefore part of any proposed measurement of an edge gap distribution. The usual low frequency noise measurement is a transport observable involving many modes. The identity \eqref{eq:main-deficit} explains why it nevertheless accesses $D_0$.

\section{Other waveforms and the outer spectral edge}\label{sec:waveforms}

For a smooth real periodic phase $Af(\varphi)$ the sideband amplitudes are Fourier coefficients
\begin{equation}\label{eq:general-phase}
 p_n(A)=\frac1{2\pi}\int_{-\pi}^{\pi}
              e^{i[n\varphi-Af(\varphi)]}\dd\varphi.
\end{equation}
Their stationary points satisfy $n/A=f'(\varphi)$. A generic outer upper edge is consequently located at $n=A\max f'$. If this maximum is attained at a unique nondegenerate point $\varphi_0$, then
\begin{equation}\label{eq:fold-conditions}
 f''(\varphi_0)=0,\qquad f'''(\varphi_0)<0.
\end{equation}
Expansion of the phase gives a cubic term and the local scale
\begin{equation}\label{eq:general-scale}
 \kappa_f=\left(\frac{A|f'''(\varphi_0)|}{2}\right)^{1/3}.
\end{equation}
This explains the Airy approximation of the individual amplitudes at a generic fold~\cite{Olver1974,BerryUpstill1980}. A kernel theorem for a specified waveform additionally needs control of all contributing stationary points and of the summed tails. We have supplied those estimates for the monochromatic case.

An outer edge imposes a useful restriction on possible degeneracies. At a smooth local maximum of $f'$, its first nonzero derivative, when it exists, has even order. The corresponding phase in \eqref{eq:general-phase} therefore has odd order. In particular, a quartic phase with $f''''(\varphi_0)\ne0$ and $f''(\varphi_0)=f'''(\varphi_0)=0$ describes a stationary inflection of $f'$, not its outer maximum. The appearance of a quartic diffraction integral at some parameters would also not, by itself, establish a Pearcey correlation kernel~\cite{Pearcey1946,TracyWidomPearcey}.

A concrete example with two tones is
\begin{equation}\label{eq:two-tone}
 f(\varphi)=\sin\varphi-\frac18\sin2\varphi,
 \qquad f'(\varphi)=\frac34-\frac12(\cos\varphi-1)^2.
\end{equation}
The upper edge is at $n=3A/4$. Its unique maximum modulo $2\pi$ is quartically flat, and
\begin{equation}\label{eq:quintic-phase}
 n\varphi-Af(\varphi)
 =(n-3A/4)\varphi+\frac A{40}\varphi^5+O(A\varphi^7).
\end{equation}
The resulting sideband width is of order $(A/8)^{1/5}$ and the local phase is quintic. This elementary example is sufficient to show why adding a second tone does not automatically turn the outer Airy edge into a Pearcey edge. Determining the full occupation kernel at such a degenerate edge, or at an internal merger of stationary points, requires a separate analysis. No Pearcey or higher-order point process limit is asserted here.

\section{Discussion}\label{sec:discussion}

The exact Bessel source kernel and its Airy limit connect a familiar conductor under a phase drive to the same fermionic edge structure found in propagating fronts. The transport consequence is especially simple: the deficit of the noise slope is an exact cumulative Bessel weight. Its limiting curve, finite amplitude correction, and finite temperature broadening can all be obtained without fitting an edge profile.

For a zero temperature comparison, the noise slope plateau centers implement the shift by half a sideband automatically. At finite temperature the explicit crossover \eqref{eq:noise-crossover} replaces the staircase by a smooth curve. The extra convolution relative to the source occupation is necessary because partition noise involves two reservoirs. These predictions require an independently characterized ac amplitude, stable electron temperature during the bias scan, and transmission that is approximately constant over the full energy window. Existing photo-assisted noise measurements establish the underlying sideband mechanism~\cite{Schoelkopf1998,Reydellet2003}. They do not by themselves establish the large-$A$ scaling studied here.

The noise deficit tests the Airy diagonal. Correlations such as \eqref{eq:scaled-covariance} and gap probabilities such as \eqref{eq:gap} are more detailed predictions about selected fermionic modes, with different requirements for energy resolution and counting time. A direct experimental comparison should specify which of these observables is measured. Extensions to nonmonochromatic voltages can then be assessed from their actual scattering amplitudes and stationary points, rather than from the number of drive parameters alone.

\bigskip
\noindent\textbf{Acknowledgments.}
The author is grateful to Leonardo Santilli and David Trillo for correspondence. This work was supported by the Shanghai Institute for Mathematics and Interdisciplinary Sciences (Grant No.~SIMIS-ID-2025-QT). ChatGPT was used during the revision of this paper.

\appendix
\section{Estimates for the Bessel kernel}\label{app:airy}

An integrable bound on the Bessel amplitudes controls both the local edge expansion and convergence of the gap determinants.

\subsection{An integrable turning-point estimate}\label{app:estimate}

We give a contour proof of a bound sufficient for the limits and the correction above. Put $\kappa=(A/2)^{1/3}$ and $\sigma=(n-A)/\kappa$. For integer $n$,
\begin{equation}\label{eq:bessel-integral}
 J_n(A)=\frac1{2\pi}\int_{-\pi}^{\pi}
                e^{i(nz-A\sin z)}\dd z.
\end{equation}
Shift the contour to $\operatorname{Im}z=1/\kappa$. The vertical sides cancel by periodicity. After the change of variable $z=(t+i)/\kappa$, the modulus of the integrand is
\begin{equation}\label{eq:contour-bound}
 \exp\!\left[-\frac n\kappa
       +A\sinh(1/\kappa)\cos(t/\kappa)\right]
 \leq C e^{-\sigma-c t^2},\qquad |t|\leq\pi\kappa,
\end{equation}
with positive constants independent of $n$ and sufficiently large $A$. Here $1-\cos u\geq2u^2/\pi^2$ for $|u|\leq\pi$ and $A\sinh(1/\kappa)-A/\kappa$ stays bounded. It follows that
\begin{equation}\label{eq:bessel-tail}
 |\kappa J_n(A)|\leq C e^{-\sigma}.
\end{equation}

The phase on the shifted contour has the expansion, writing $w=t+i$,
\begin{equation}\label{eq:contour-expansion}
 nz-A\sin z=\sigma w+\frac{w^3}{3}
             -\frac{w^5}{60\kappa^2}+O\!\left(\frac{|w|^7}{\kappa^4}\right).
\end{equation}
Taylor expansion on a central portion of the contour, with the Gaussian bound \eqref{eq:contour-bound} controlling its complement, gives
\begin{equation}\label{eq:bessel-expanded}
 \kappa J_n(A)=\Ai(\sigma)
 -\frac{\Ai^{(5)}(\sigma)}{60\kappa^2}
 +O\!\left(\kappa^{-4}e^{-\sigma}\right).
\end{equation}
This estimate is uniform for integer $n$ with $\sigma\geq-M$, for each fixed $M$. To justify the error explicitly, one can split $|t|\leq\kappa^\epsilon$ and its complement, with a fixed sufficiently small $\epsilon>0$. On the central part the Taylor remainder is a polynomial in $|t|$ times $\kappa^{-4}e^{-\sigma-c't^2}$. On the complement both the exact integrand and its Airy expansion have a Gaussian integral smaller than any power of $\kappa^{-1}$, with the same factor $e^{-\sigma}$. The constants are uniform in $\sigma$, which appears only through the oscillatory factor $e^{i\sigma t}$ and the modulus $e^{-\sigma}$.

For the present purposes the consequence
\begin{equation}\label{eq:integrable-remainder}
 |\kappa J_n(A)-\Ai(\sigma)|
 \leq C_M\kappa^{-2}e^{-\sigma},\qquad \sigma\geq-M,
\end{equation}
is enough. The Airy derivative in \eqref{eq:bessel-expanded} has the same exponential bound on this half-line, as follows from its shifted contour integral. Formula~\eqref{eq:bessel-expanded} is consistent with the standard transition expansion~\cite{Olver1974,DLMF}. It also gives an independent way to check the sign of the first correction. For example, $\Ai^{(5)}(s)=4s\Ai(s)+s^2\Ai'(s)$.

\subsection{Kernel and determinant convergence}

For bounded $s,t$, \eqref{eq:integrable-remainder} applies to all $\nu_A(s)+r$ and $\nu_A(t)+r$ with $r\geq0$. Thus
\[
 \kappa B_A(\nu_A(s),\nu_A(t))
 =\frac1\kappa\sum_{r=0}^\infty
  \Ai\!\left(\frac{\nu_A(s)-A+r}{\kappa}\right)
  \Ai\!\left(\frac{\nu_A(t)-A+r}{\kappa}\right)
 +O(\kappa^{-2}).
\]
The integrable exponential envelope permits uniform passage from the sum to the integral. The floor contributes an argument shift of at most $1/\kappa$, proving \eqref{eq:kernel-limit}. Applying the first Euler--Maclaurin boundary term to this expression proves Proposition~\ref{prop:correction}.

For the trace-norm statement, embed each lattice site into a cell of length $1/\kappa$ and define the factor kernel
\begin{equation}\label{eq:factor}
 \Phi_A(s,u)=\kappa J_{\lfloor A+\kappa s\rfloor+\lfloor\kappa u\rfloor}(A),
 \qquad u\geq0.
\end{equation}
On a right half-line starting at the cell associated with $N_A$, the embedded rescaled Bessel kernel is exactly $\Phi_A\Phi_A^*$. Its limiting factor is $\Phi(s,u)=\Ai(s+u)$. The bound \eqref{eq:bessel-tail} implies an integrable square envelope $C e^{-2(s+u)}$ for $s\geq s_0-1$ and $u\geq0$. Dominated convergence therefore gives $\|\Phi_A-\Phi\|_2\to0$, where $\|\cdot\|_2$ is the Hilbert--Schmidt norm. Moving the first cell endpoint by $O(1/\kappa)$ does not change that limit. Finally,
\begin{equation}\label{eq:trace-norm}
 \|\Phi_A\Phi_A^*-\Phi\Phi^*\|_1
 \leq(\|\Phi_A\|_2+\|\Phi\|_2)\|\Phi_A-\Phi\|_2\longrightarrow0.
\end{equation}
Continuity of the Fredholm determinant in trace norm proves \eqref{eq:gap} and, after multiplication of the operators by $z-1$, also convergence of the generating functions in \eqref{eq:fcs} for each fixed $z$.

\section{Thermal convolution identities and limits}

Two Fermi occupations determine the exact noise broadening, and uniform bounds justify its limit at the spectral edge.

\subsection{Two Fermi occupations}\label{app:thermal-identity}

Use energy units in which the temperature is $\vartheta$, and let $f_\vartheta(y)=(1+e^{y/\vartheta})^{-1}$. Since $q_\vartheta=-f_\vartheta'$, one can write
\begin{equation}\label{eq:fermi-mixture}
 f_\vartheta(y)=\int_\R q_\vartheta(a)\one_{y<a}\dd a.
\end{equation}
The phase-space integral for a particle exchanged between two reservoirs is
\begin{align}\label{eq:phase-space}
 h_\vartheta(y)
 &:={\int_\R}\left[
 f_\vartheta(\epsilon-y)(1-f_\vartheta(\epsilon))
 +f_\vartheta(\epsilon)(1-f_\vartheta(\epsilon-y))\right]\dd\epsilon\nonumber\\
 &=y\coth\!\left(\frac{y}{2\vartheta}\right).
\end{align}
The value at zero is $2\vartheta$. Inserting \eqref{eq:fermi-mixture} first turns the integral over $\epsilon$ into the sum of two interval lengths,
\begin{equation}\label{eq:logistic-difference}
 h_\vartheta(y)
 =\int_{\R^2}q_\vartheta(a)q_\vartheta(b)|y+a-b|\dd a\dd b
 =(q_\vartheta*q_\vartheta*|\cdot|)(y).
\end{equation}
Both densities are even. Differentiating twice gives
\begin{equation}\label{eq:second-derivative}
 \frac12h_\vartheta''=q_\vartheta*q_\vartheta=w_\vartheta,
 \qquad h_\vartheta'=g_\vartheta.
\end{equation}
This proves the closed form \eqref{eq:w}, including its removable value at zero. It also proves positivity and normalization without cancellation of singular hyperbolic terms. Convolving the zero temperature step $\one_{x<n}$ with $w_\vartheta$ gives $[1-g_\vartheta(x-n)]/2$. Summing with weights $J_n(A)^2$ proves Proposition~\ref{prop:thermal-convolution}.

The Fourier transform of $q_\theta$ is $\pi\theta k/\sinh(\pi\theta k)$, with value one at $k=0$. Its expansion gives $\int v^2q_\theta(v)\dd v=\pi^2\theta^2/3$. The convolution adds the variances. Since both densities are even and have finite moments of every order, Taylor expansion of \eqref{eq:source-convolution} and \eqref{eq:noise-crossover} gives \eqref{eq:source-sommerfeld}--\eqref{eq:noise-sommerfeld}, using $D_0''=-2\Ai\Ai'$.

\subsection{Domination in the edge limit}\label{app:thermal-limit}

Besides \eqref{eq:bessel-tail}, a uniform bound for oscillatory integrals gives
\begin{equation}\label{eq:global-bessel}
 |J_n(A)|\leq C A^{-1/3},\qquad n\in\Z,\quad A\geq1.
\end{equation}
For completeness, split the circle in \eqref{eq:bessel-integral} into a fixed number of intervals on which either $|\sin z|$ or $|\cos z|$ is bounded below. The second or third derivative of the phase is then bounded below by a constant times $A$. The corresponding van der Corput bounds are uniform in the linear term $nz$ and give \eqref{eq:global-bessel}. This is also consistent with the uniform Bessel estimates in Refs.~\cite{Olver1974,DLMF}.

In \eqref{eq:source-thermal}, the rescaled product of Bessel functions is uniformly bounded by \eqref{eq:global-bessel}. For negative $r/\kappa$, the Fermi factor decays exponentially, uniformly when $\vartheta/\kappa$ stays in a compact subset of $(0,\infty)$. For positive $r/\kappa$, \eqref{eq:bessel-tail} supplies exponential decay. Dominated convergence of the Riemann sums proves \eqref{eq:source-thermal-limit}.

For the noise limit, these two bounds also imply
\begin{equation}\label{eq:global-diagonal}
 \kappa_A d_{A,0}(A+\kappa_A s)
 \leq C(1+|s|),\qquad s\in\R,
\end{equation}
with a constant independent of sufficiently large $A$. To see this, bound each summand with order between the threshold and $A$ by $C/\kappa_A^2$, and use the exponential bound above $A$. There are at most $\kappa_A(-s)_++O(1)$ terms of the first kind. This upper bound remains valid even for thresholds far to the left.

Rescaling \eqref{eq:exact-convolution} now gives the exact formula
\[
 \kappa_A d_{A,\vartheta}(A+\kappa_A s)
 =\int_\R w_{\theta_A}(v)\,
     \kappa_A d_{A,0}(A+\kappa_A(s-v))\dd v.
\]
The linear bound \eqref{eq:global-diagonal} and the exponential tails of $w_{\theta_A}$ justify locally uniform dominated convergence for $\theta_A\to\theta>0$. This proves Theorem~\ref{thm:thermal-noise}. The same bound and the approximate-identity property of $q_\theta$ and $w_\theta$ give their zero temperature limits.

\section{Numerical evaluation}\label{app:numerics}

The figures and Table~\ref{tab:errors} are generated by the accompanying script \texttt{verify\_and\_plot.py}. It evaluates Bessel functions directly and forms the sums in \eqref{eq:bessel-kernel} and \eqref{eq:thermal-deficit}. No curve parameters are fitted. The sums retain orders through at least $A+45\kappa_A+120$, where the tail is negligible for the displayed amplitudes. The infinite Airy convolutions use adaptive quadrature, with the thermal variable restricted to $48$ temperature units on either side. Their exponential tails are negligible at the displayed precision. Numerical values in the table are rounded to four significant digits.

Independent checks compare direct sums with \eqref{eq:cd}, numerical bias derivatives with \eqref{eq:noise-slope}, and convolution of two densities $q_\theta$ with the closed form \eqref{eq:w}. The script also compares the exact finite temperature noise with convolution of the exact staircase, and checks the low temperature coefficients and the Bessel correction \eqref{eq:bessel-expanded}. These checks support the analytical derivations. The convergence theorems do not rely on the numerical results.

% Bibliography is inserted below in alphabetical order.
\end{document}